\documentclass{amsart}

\usepackage{amssymb, amscd,txfonts,stmaryrd,setspace,mathabx,mathtools,makecell}
\usepackage{enumerate}
\usepackage[usenames,dvipsnames]{xcolor}
\usepackage[bookmarks=true,colorlinks=true, linkcolor=MidnightBlue, citecolor=cyan]{hyperref}
\usepackage{lmodern}
\usepackage{graphicx,float}

\input xy
\xyoption{all} \xyoption{poly} \xyoption{knot}\xyoption{curve}
\usepackage{xypic,color}
\newcommand{\comment}[1]{}

\newcommand{\longnote}[2][4.9in]{\fcolorbox{black}{yellow}{\parbox{#1}{\color{black} #2}}}

\def\tn{\textnormal}

\def\mc{\mathcal}

\def\ZZ{{\mathbb Z}}

\def\RR{{\mathbb R}}
\def\CC{{\mathbb C}}

\def\PP{{\mathbb P}}
\def\NN{{\mathbb N}}

\def\Hom{\tn{Hom}}

\def\Fun{\tn{Fun}}

\def\Ob{\tn{Ob}}

\def\SEL*{\tn{SEL*}}

\def\hsp{\hspace{.3in}}

\def\singleton{{\{*\}}}
\def\Loop{{\mcL oop}}
\def\LoopSchema{{\parbox{.5in}{\fbox{\xymatrix{\LMO{s}\ar@(l,u)[]^f}}}}}

\def\to{\rightarrow}
\def\from{\leftarrow}

\def\cross{\times}
\def\taking{\colon}

\def\too{\longrightarrow}
\def\fromm{\longleftarrow}

\def\ss{\subseteq}

\def\iso{\cong}

\def\|{{\;|\;}}
\def\m1{{-1}}
\def\op{^\tn{op}}

\def\wt{\widetilde}

\def\ol{\overline}

\newcommand{\LMO}[1]{\stackrel{#1}{\bullet}}
\newcommand{\LTO}[1]{\stackrel{\tt{#1}}{\bullet}}
\newcommand{\LA}[2]{\ar[#1]^-{\tn {#2}}}

\def\ullimit{\ar@{}[rd]|(.3)*+{\lrcorner}}
\def\urlimit{\ar@{}[ld]|(.3)*+{\llcorner}}
\def\lllimit{\ar@{}[ru]|(.3)*+{\urcorner}}
\def\lrlimit{\ar@{}[lu]|(.3)*+{\ulcorner}}
\def\ulhlimit{\ar@{}[rd]|(.3)*+{\diamond}}
\def\urhlimit{\ar@{}[ld]|(.3)*+{\diamond}}
\def\llhlimit{\ar@{}[ru]|(.3)*+{\diamond}}
\def\lrhlimit{\ar@{}[lu]|(.3)*+{\diamond}}
\newcommand{\clabel}[1]{\ar@{}[rd]|(.5)*+{#1}}
\newcommand{\TriRight}[7]{\xymatrix{#1\ar[dr]_{#2}\ar[rr]^{#3}&&#4\ar[dl]^{#5}\\&#6\ar@{}[u] |{\Longrightarrow}\ar@{}[u]|>>>>{#7}}}
\newcommand{\TriLeft}[7]{\xymatrix{#1\ar[dr]_{#2}\ar[rr]^{#3}&&#4\ar[dl]^{#5}\\&#6\ar@{}[u] |{\Longleftarrow}\ar@{}[u]|>>>>{#7}}}
\newcommand{\TriIso}[7]{\xymatrix{#1\ar[dr]_{#2}\ar[rr]^{#3}&&#4\ar[dl]^{#5}\\&#6\ar@{}[u] |{\Longleftrightarrow}\ar@{}[u]|>>>>{#7}}}

\newcommand{\arr}[1]{\ar@<.5ex>[#1]\ar@<-.5ex>[#1]}
\newcommand{\arrr}[1]{\ar@<.7ex>[#1]\ar@<0ex>[#1]\ar@<-.7ex>[#1]}
\newcommand{\arrrr}[1]{\ar@<.9ex>[#1]\ar@<.3ex>[#1]\ar@<-.3ex>[#1]\ar@<-.9ex>[#1]}
\newcommand{\arrrrr}[1]{\ar@<1ex>[#1]\ar@<.5ex>[#1]\ar[#1]\ar@<-.5ex>[#1]\ar@<-1ex>[#1]}

\newcommand{\To}[1]{\xrightarrow{#1}}
\newcommand{\Too}[1]{\xrightarrow{\ \ #1\ \ }}
\newcommand{\From}[1]{\xleftarrow{#1}}

\newcommand{\Adjoint}[4]{\xymatrix@1{{#2}\ar@<.5ex>[r]^-{#1} &{#3} \ar@<.5ex>[l]^-{#4}}}

\def\id{\tn{id}}

\def\Cat{{\bf Cat}}
\def\Monad{{\bf Monad}}

\def\Vect{\text{Vect}}
\def\Rep{{\bf Rep}}

\def\List{\tn{List}}
\def\Exc{\tn{Exc}}

\def\Rel{{\bf Rel}}

\def\Type{{\bf Type}}
\def\Set{{\bf Set}}

\def\set{{\text \textendash}{\bf Set}}

\def\bhline{\Xhline{2\arrayrulewidth}}
\def\bbhline{\Xhline{2.5\arrayrulewidth}}

\def\colim{\mathop{\tn{colim}}}

\def\mcC{\mc{C}}
\def\mcD{\mc{D}}

\def\mcI{\mc{I}}

\def\mcK{\mc{K}}
\def\mcL{\mc{L}}

\def\mcP{\mc{P}}

\def\mcS{\mc{S}}

\def\mcW{\mc{W}}
\def\mcX{\mc{X}}

\newcommand{\subsub}[1]{\setcounter{subsubsection}{\value{theorem}}\subsubsection{#1}\addtocounter{theorem}{1}}

\newcommand{\back}[1]{\stackrel{\from}{#1}\!}
\newcommand{\Kls}[1]{{\bf Kls}({#1})}
\newcommand{\bkls}[1]{{\text \textendash}\Kls{#1}}
\newcommand{\kls}[1]{{\text \textendash}\wt{\bf Kls}({#1})}
\def\Dist{\text{Dist}}
\def\Supp{{\bf Supp}}
\def\Inp{\tn{Inp}}
\def\Tur{\tn{Tur}}
\def\Halt{\{\text{Halt}\}}
\def\Tape{{T\!ape}}

\def\Sch{{\bf Sch}}
\def\Fin{{\bf Fin}}

\newtheorem{theorem}{Theorem}[subsection]
\newtheorem{lemma}[theorem]{Lemma}
\newtheorem{proposition}[theorem]{Proposition}

\theoremstyle{remark}
\newtheorem{remark}[theorem]{Remark}
\newtheorem{example}[theorem]{Example}

\newtheorem{question}[theorem]{Question}
\newtheorem{guess}[theorem]{Guess}

\theoremstyle{definition}
\newtheorem{definition}[theorem]{Definition}
\newtheorem{notation}[theorem]{Notation}

\newcommand{\mainCatLarge}[1]{ 
	\stackrel{#1}{
		\parbox{4.5in}{\fbox{\parbox{4.4in}{\begin{center}\underline{{\tt Employee} manager worksIn $\simeq$ {\tt Employee} worksIn}\hsp  \underline{{\tt Department} secretary worksIn $\simeq$ {\tt Department}}\end{center}~\\\\\\
			\xymatrix@=8pt{&\LTO{Employee}\ar@<.5ex>[rrrrr]^{\tn{worksIn}}\ar@(l,u)[]+<5pt,10pt>^{\tn{manager}}\ar[dddl]_{\tn{first}}\ar[dddr]^{\tn{last}}&&&&&\LTO{Department}\ar@<.5ex>[lllll]^{\tn{secretary}}\ar[ddd]^{\tn{name}}\\\\\\\LTO{FirstNameString}&&\LTO{LastNameString}&~&~&~&\LTO{DepartmentNameString}
			}
		}}}
	}
}

\begin{document}

\title{Kleisli database instances}

\author{David I. Spivak}

\address{Department of Mathematics, Massachusetts Institute of Technology, Cambridge MA 02139}

\email{dspivak@gmail.com}

\thanks{This project was supported by ONR grant N000141010841.}

\begin{abstract}

We use monads to relax the atomicity requirement for data in a database. Depending on the choice of monad, the database fields may contain generalized values such as lists or sets of values, or they may contain exceptions such as various types of nulls. The return operation for monads ensures that any ordinary database instance will count as one of these generalized instances, and the bind operation ensures that generalized values behave well under joins of foreign key sequences. Different monads allow for vastly different types of information to be stored in the database. For example, we show that classical concepts like Markov chains, graphs, and finite state automata are each perfectly captured by a different monad on the same schema. 

\end{abstract}

\maketitle

\tableofcontents

\section{Introduction}

Monads are category-theoretic constructs with wide-ranging applications in both mathematics and computer science. In \cite{Mog}, Moggi showed how to exploit their expressive capacity to incorporate fundamental programming concepts into purely functional languages, thus considerably extending the potency of the functional paradigm. Using monads, concepts that had been elusive to functional programming, such as state, input/output, and concurrency, were suddenly made available in that context.

In the present paper we describe a parallel use of monads in databases. This approach stems from a similarity between categories and database schemas, as presented in \cite{Sp1}. The rough idea is as follows. A database schema can be modeled as a category $\mcC$, and an ordinary database instance is a functor $\delta\taking\mcC\to\Set$. Given a monad $T\taking\Set\to\Set$, a {\em Kleisli $T$-instance} is a functor $$\delta\taking\mcC\to\Kls{T},$$ where $\Kls{T}$ is the Kleisli category of $T$, as will be explained in Section \ref{sec:monads and Kleisli}.

Values in a Kleisli $T$-instance are less restricted than ordinary values; we call these generalized values {\em $T$-values}. In particular, within a Kleisli instance we are permitted to relax the atomicity requirement for data (a requirement found in Codd's notion of {\em first normal form}, see \cite{Cod}), while still maintaining referential integrity. For example, if $T$ is the List monad then $T$-values are lists, so a single entry in a foreign key or data column could contain a list of entries of the target type. Similarly, $T$-values might include assurance information (i.e. a number between 0\% and 100\%), in which case each datum would come equipped with a probability of correctness. Importantly, the monadicity of $T$ ensures that the extra information in $T$-values will naturally and predictably synthesize along any path obtained by joining a sequence of foreign keys. One can think of flattening lists of lists, of multiplying probabilities, or of propagating exceptions.

Kleisli instances offer additional functionality in a database, and such functionalities vary widely as the category of monads on $\Set$ is quite rich. Having a variety of available possibilities, the database architect can choose those that best fit the current needs. Moreover, a morphism between monads $T\to T'$ induces a functor from the category of $T$-instances to the category of $T'$-instances on the schema. In future work we will show that one can vary the choice of monad throughout the database schema, thus greatly increasing the expressive power of database schemas. By incorporating these features within the design specification of the database, as opposed to applying them from without, we reduce the barrier between database and program. Whereas normally such functionality is distributed throughout the technology stack, the monadic approach leads to a centralization of features, increasing our ability to manage the system with certainty.

The monad formalism also enables more economical schema design. For example, typically one encodes a set-membership relation with three tables, e.g. $$\fbox{$\LTO{element}\fromm\LTO{membership}\too\LTO{set}$},$$ and to encode that A has as an attribute a list of B's requires requires even more overhead. However, the same information can be captured with a single column when one employs the Multiset or List monad. 

As an aside, the monad formalism also yields a surprising coincidence. We show that there is a database schema $\Loop$ such that, for different choices of monads $T$, the set of Kleisli $T$-instances on $\Loop$ can be interpreted in terms of classical mathematical subjects. 
\begin{align}\label{dia:riddle}
\begin{tabular}{| l l l |}
\bhline
{\bf Classical mathematical subject}&{\bf monad $T$}&{\bf Internal reference}\\\bbhline
Discrete dynamical systems & Atomic & Example \ref{ex:loop}\\\hline
Graphs & Multiset & Example \ref{ex:graphs}\\\hline
Markov chains & $\Dist$ & Example \ref{ex:markov}\\\hline
Finite state automata & $\Inp^U$ & Example \ref{ex:fsa}\\\hline
Turing machines & $\Tur^{\{0,1\}}_{\{L,R,W_0,W_1\}}$& Example \ref{ex:turing}\\\hline
\parbox{1.6in}{``Jordan Canonical Form" (vector spaces with endomorphism)}&$\Vect$&Example \ref{ex:rep theory}\\\hline
Multigraphs&Free rig&Example \ref{ex:multigraphs}\\
\bhline
\end{tabular}
\end{align}

Monads have been applied to databases in previous work (e.g. \cite{BNT}, \cite{Gru}, \cite{LT}, and \cite{Maj})), but the sense in which they are applied is totally different than that which is presented in this paper. In each of these papers, monads were applied to make sense of queries and, in particular, aggregate functions on collections (e.g. counts and sums). The present paper, on the other hand, deals with the employment of monads within the database schema to provide additional expressivity in each field, e.g. allowing non-atomic data or annotating data with probability of correctness. While previous work may simplify aggregation in our context, it should be seen as orthogonal to the ideas presented here.

In this paper we assume the reader has encountered categories before, but it is not totally necessary. Readers with either very much or very little category theory may benefit by reading Section \ref{sec:CatDB} and Example \ref{ex:monad} and then skipping directly to Section \ref{sec:examples}. Readers with some background but who wish to review monads or their Kleisli categories will hopefully be satisfied with the brief overview in the intermediate sections. For a good reference on category theory, and monads in particular, one should consult \cite{Awo} or \cite{BW}.

We begin this paper in Section \ref{sec:background} with a brief review of the categorical model of databases, as well as some background on monads and their Kleisli categories. We discuss a new application of monads to databases in Section \ref{sec:Kleisli instances}. In section \ref{sec:examples} we offer several examples that may be of interest, such as the List-instances. In Section \ref{sec:transformations} we discuss morphisms of monads, which for example allow one to transform ordinary atomic instances into List-instances. Finally in Section \ref{sec:future} we briefly discuss our plans for future work in this area.

\subsection{Acknowledgements}

I'd like to thank Steve Lack and Tom Leinster for their excellent answers to a question I posted on \href{http://mathoverflow.net/questions/55182/what-is-known-about-the-category-of-monads-on-set}{\texttt mathoverflow.net}, and I'd like to thank Allen Brown, Peter Gates, and Ka Yu Tam for many useful conversations.

\section{Background}\label{sec:background}

In this section we recount a simple category-theoretic model of databases, and then review basic material on monads.

\begin{notation}

Let $\Set$ denote the category whose objects are sets and whose morphisms are functions. Throughout the paper we will be careful to reserve the word {\em function} to refer to mappings between sets. In a general category $\mcK$ we use words like {\em arrow} or {\em map}, but never function, to refer to morphisms in $\mcK$. 

\end{notation}

\subsection{Categorical databases}\label{sec:CatDB}

We begin with some background on so-called categorical databases. Much more can be found in \cite{Sp1}.

Roughly, a database {\em schema} is a category presentation: it is given by a set of objects (which will be drawn as nodes), a set of generating arrows, and an equivalence relation on paths. We denote a path by writing its source object followed by a sequence of arrows. We denote an equivalence of paths using the $\simeq$-sign. For example, consider the following schema:
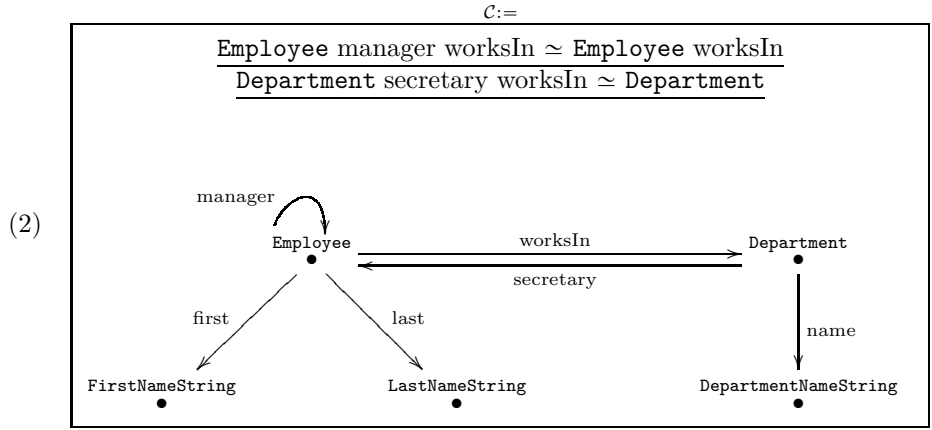
\begin{align}\label{dia:mainCatLarge}
\mainCatLarge{\mcC:=}
\end{align}
Here we see a graph with five vertices and six arrows, and underlined at the top we see two path equivalence (PE) statements.\footnote{The first PE statement, ``{\tt Employee} manager worksIn $\simeq$ {\tt Employee} worksIn", identifies a path of length 2 with a path of length 1. The second PE statement, ``{\tt Department} secretary worksIn $\simeq$ {\tt Department}", identifies a path of length 2 with a path of length 0.} This information generates a category: the free category on the graph, modulo the path equivalence relation. In fact, in \cite[3.4.1]{Sp1} the author defines a category $\Sch$, whose objects are schemas (presented categories) as above, and proves that $\Sch$ is equivalent to $\Cat$. From here on, we elide the difference between a schema (category presentation) and a category.

A schema $\mcC$ is supposed to describe the wiring of a database. We think of each object $c\in\Ob(\mcC)$ as representing a table and each arrow $f\taking c\to c'$ emanating from $c$ as representing a column of $c$ that takes values in table $c'$. Roughly, an {\em instance} on $\mcC$ is the actual data: more precisely, an instance assigns to each table a set of rows of data that conform to the specifications given by $\mcC$. For example, the schema represented in Diagram (\ref{dia:mainCatLarge}) describes the wiring of the following database instance: 
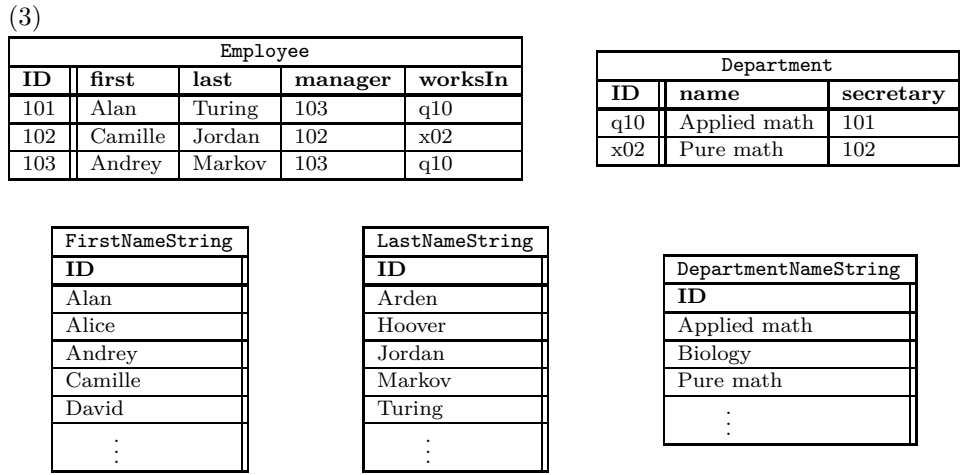
\begin{align}\label{dia:instance on maincat}
&\footnotesize
\begin{tabular}{| l || l | l | l | l |}\bhline
\multicolumn{5}{| c |}{{\tt Employee}}\\\bhline 
{\bf ID}&{\bf first}&{\bf last}&{\bf manager}&{\bf worksIn}\\\bbhline 101&Alan&Turing&103&q10\\\hline 102&Camille&Jordan&102&x02\\\hline 103&Andrey&Markov&103&q10\\\bhline
\end{tabular}&\hspace{.25in}\footnotesize
\begin{tabular}{| l || l | l |}\bhline
\multicolumn{3}{| c |}{{\tt Department}}\\
\bhline {\bf ID}&{\bf name}&{\bf secretary}\\\bbhline q10&Applied math&101\\\hline x02&Pure math&102\\\bhline
\end{tabular}
\end{align}\vspace{.1in}
\begin{align*}\footnotesize
\begin{tabular}{| l ||}\bhline
\multicolumn{1}{| c |}{{\tt FirstNameString}}\\\bhline
{\bf ID}\\\bbhline Alan\\\hline Alice\\\hline Andrey\\\hline Camille\\\hline David\\\hline\hspace{.25in}\vdots\\\bhline
\end{tabular}\hspace{.6in}\footnotesize
\begin{tabular}{| l ||}\bhline
\multicolumn{1}{| c |}{{\tt LastNameString}}\\\bhline
{\bf ID}\\\bbhline Arden\\\hline Hoover\\\hline Jordan\\\hline Markov\\\hline Turing\\\hline\hspace{.25in}\vdots\\\bhline
\end{tabular}\hspace{.6in}\footnotesize
\begin{tabular}{| l ||}\bhline
\multicolumn{1}{| c |}{{\tt DepartmentNameString}}\\\bhline
{\bf ID}\\\bbhline Applied math\\\hline Biology\\\hline Pure math\\\hline\hspace{.25in}\vdots\\\bhline
\end{tabular}
\end{align*}
Every table has an ID column and perhaps other columns. Counting tables in (\ref{dia:instance on maincat}) we find five, the number of nodes in (\ref{dia:mainCatLarge}); and counting the non-ID columns in (\ref{dia:instance on maincat}) we find six, the number of arrows in (\ref{dia:mainCatLarge}). 

In fact, we can see that Diagram (\ref{dia:instance on maincat}) constitutes an assignment of a set (of rows) to each node in $\mcC$ and a function to each arrow in $\mcC$. For example the node {\tt Employee} is assigned the set $\{101,102,103\}$ and the arrow {\bf manager}$\taking{\tt Employee}\to{\tt Employee}$ is assigned the function sending $101\mapsto 103$ and $102\mapsto102$ and $103\mapsto 103$. Thus an instance on schema $\mcC$ is precisely a functor $I\taking\mcC\to\Set$. The path equivalence relation on $\mcC$ ensures that the values behave in specified ways. For example, Alan Turing's manager is Andrey Markov, and these two men (are required to) work in the same department, Applied math. Similarly, the Secretary of Pure math is Camille Jordan, and he (by necessity) works in the Pure math department.

We summarize all this in a formal, if hasty, definition. A careful description is given in \cite{Sp1}.

\begin{definition}

A {\em schema} $\mcC$ is a small category presentation. An {\em instance on $\mcC$} is a functor $\delta\taking\mcC\to\Set$.

\end{definition}

Throughout this paper we will continually return to a couple examples.

\begin{example}\label{ex:atomic worksIn}

One of the most basic categories is the so-called free-arrow category. We simply add names to make it more reminiscent of databases.
$$\mcW:=\fbox
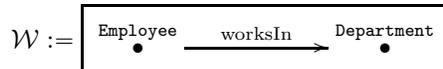
{\xymatrix{\LTO{Employee}\LA{rr}{worksIn}&&\LTO{Department}}}$$

A instance $\delta\taking\mcW\to\Set$ consists of a set of employees, a set of departments, and a function mapping each employee to a department. For example 
$$\delta:=\left\{\hsp
\begin{tabular}{| l || l |}\bhline
\multicolumn{2}{| c |}{\tt{Employee}}\\\bhline 
{\bf ID}&{\bf WorksIn}\\\bbhline
Alice & Math\\\hline
Bob & EECS\\\hline
Carl & Ling\\\hline
Deb & EECS\\\hline
Fred & Math\\\hline
Jen & Bio\\\bhline
\end{tabular}
\hspace{.5in}
\begin{tabular}{| l ||}\bhline
\multicolumn{1}{| c |}{\tt{Department}}\\\bhline 
{\bf ID}\\\bbhline
Bio\\\hline
EECS\\\hline
Ling\\\hline
Math\\\hline
Music\\\bhline
\end{tabular}
\hsp\right\}
$$

\end{example}

\begin{example}\label{ex:loop}

The schema represented here $$\Loop:=\LoopSchema$$
has one object $s$ and one generating arrow $f$, but a countably infinite set $\{f^n\|n\in\NN\}$ of paths. An instance $\delta\taking\Loop\to\Set$ is often called a {\em discrete dynamical system}. It consists of a set, which we might think of as the set of states of the system, and a function from that set to itself, which we might think of as the ``next state" function. As with any database instance we can apply the Grothendieck construction (see \cite{Sp2}) and get a nice picture of the system. For example one might have
$$\delta:=
\tiny\begin{tabular}{| l || l |}\bhline
\multicolumn{2}{| c |}{\tt{s}}\\\bhline 
{\bf ID}&{\bf f}\\\bbhline
A & B\\\hline
B & C\\\hline
C & C\\\hline
D & B\\\hline
E & C\\\hline
F & G\\\hline
G & H\\\hline
H & G\\\hline
\end{tabular}
\hsp
\xmapsto{\;\;\text{pictured}\;\;}\hsp
\int\delta=\parbox{1.4in}{\fbox
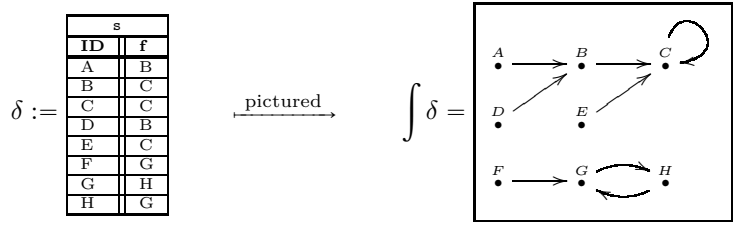
{\xymatrix@=7pt{
\LMO{A}\ar[rr]&&\LMO{B}\ar[rr]&&\LMO{C}\ar@(u,r)[]^~\\
\LMO{D\ar[urr]}&&\LMO{E}\ar[urr]\\
\LMO{F}\ar[rr]&&\LMO{G}\ar@/^.5pc/[rr]&&\LMO{H}\ar@/^.5pc/[ll]^~
}}}
$$

Discrete dynamical systems are commonly used in modeling \cite{San}. In fact, we will see throughout this paper that Kleisli instances on $\Loop$ are equivalent to structures of classical mathematical interest. A list of such examples is provided in the Introduction, Table (\ref{dia:riddle}).

\end{example}

\subsection{Monads and Kleisli categories}\label{sec:monads and Kleisli}

In this section we define monads on the category $\Set$ and their Kleisli categories. Monads can be defined on any category, but the discussion will be a bit simpler if we are content with specializing to $\Set$. One can replace $\Set$ with $\Type$, the category of types for any typed $\lambda$-calculus, in what follows.

\begin{definition}\label{def:monad}

A {\em monad} $\top$ on $\Set$ consists of a triple $\top:=(T,\eta,\mu)$, where $T\taking\Set\to\Set$ is a functor and $\eta\taking\id_{\Set}\to T$ and $\mu\taking T\circ T\to T$ are natural transformations, such that the following diagrams commute: 
\begin{align}
\label{dia:monad axiom1}\xymatrix{T\circ\id_\Set\ar[r]^-{\id_T\circ\eta}\ar@{=}[dr]&T\circ T\ar[d]^\mu\\&T}\\
\label{dia:monad axiom2}\xymatrix{\id_\Set\circ T\ar[r]^-{\eta\circ\id_T}\ar@{=}[dr]&T\circ T\ar[d]^\mu\\&T}\\
\label{dia:monad axiom3}\xymatrix{T\circ T\circ T\ar[r]^-{\mu\circ\id_T}\ar[d]_-{\id_T\circ\mu}&T\circ T\ar[d]^\mu\\T\circ T\ar[r]_\mu&T}\end{align}

We call $T$ the {\em functor part} of $\top$ and we refer to $\eta$ and $\mu$ as the {\em unit map} and the {\em multiplication map} of $\top$, respectively. We sometimes abuse notation and refer to the functor part $T$ as though it were the whole monad.

\end{definition}

\begin{example}\label{ex:monad}

We now go through Definition \ref{def:monad} using the $\List$ monad. The first step is to give a functor $\List\taking\Set\to\Set$. For every set $X$ we must provide a set $\List(X)$ and for every function $f\taking X\to Y$ we must provide a function $\List(f)\taking\List(X)\to\List(Y)$. To clarify the situation, let us lay out two sets $X, Y$ and a function between them.
$$X=\{p,q,r\}\hspace{1in}Y=\{1,2,3,4\}
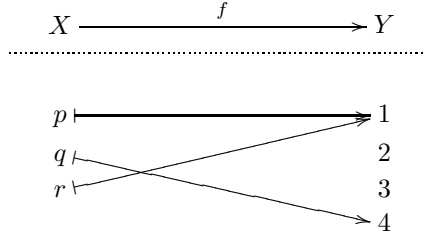
$$
$$\xymatrix@=.5pt{
\hspace{.3in}&X\ar[rr]^f&\hspace{1.4in}&Y\\
\ar@{..}[rrr]+<.25in,0pt>&&&\\\\\\
&p\ar@{|->}[rr]&&1\\
&q\ar@{|->}[rrdd]&&2\\
&r\ar@{|->}[rruu]&&3\\
&&&4}$$
Then $\List(X)$ is the set of all lists in elements of $X$. Thus the set $\List(X)$ includes the empty list $[]$, one element lists such as $[p]$, and all other lists in $X$ (of finite length) such as $[p,q,r,r,p]$. Given our function $f$ as above, we can apply it term-by-term to a list in $X$ and return a list of the same length in $Y$. 
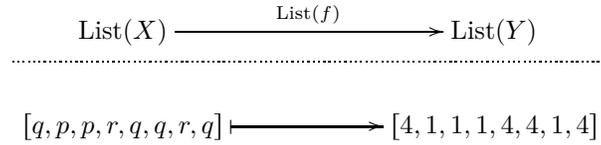
$$\xymatrix@=.5pt{
&\List(X)\ar[rr]^{\List(f)}&\hspace{.7in}&\List(Y)\\
\ar@{..}[rrr]+<.6in,0pt>&&&\\\\\\
&[q, p, p, r, q, q, r, q]\ar@{|->}[rr]&&[4,1,1,1,4, 4,1,4]}$$
Thus we have described the functor $\List$. As a monad, it comes with two natural transformations, a unit map $\eta$ and a multiplication map $\mu$. Given a set $X$, the unit map $\eta_X\taking X\to\List(X)$ returns singleton lists as follows
$$
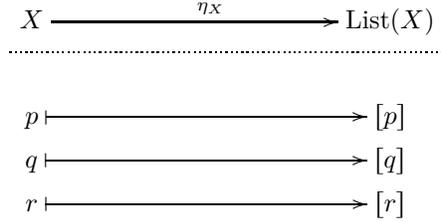
\xymatrix@=.5pt{
&X\ar[rr]^{\eta_X}&\hspace{1.4in}&\List(X)\\
\ar@{..}[rrr]+<.3in,0pt>&&&\\\\\\
&p\ar@{|->}[rr]&&[p]\\
&q\ar@{|->}[rr]&&[q]\\
&r\ar@{|->}[rr]&&[r]}$$
Given a set $X$, the multiplication map $\mu_X\taking\List(\List(X))\to\List(X)$ flattens lists of lists as follows.
$$
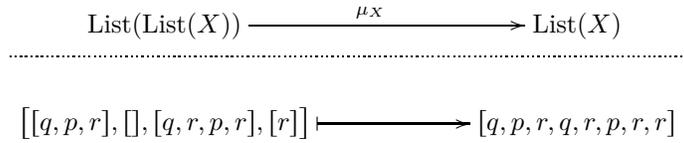
\xymatrix@=.5pt{
&\List(\List(X))\ar[rr]^{\mu_X}&\hspace{.7in}&\List(X)\\
\ar@{..}[rrr]+<.6in,0pt>&&&\\\\\\
&\big[[q, p, r], [], [q, r, p, r], [r]\big]\ar@{|->}[rr]&&[q, p, r, q, r, p, r, r]}$$
The naturality of $\eta$ and $\mu$ just mean that these maps work appropriately well under term-by-term replacement by a function $f\taking X\to Y$. Finally the three monad axioms (\ref{dia:monad axiom1}), (\ref{dia:monad axiom2}), and (\ref{dia:monad axiom3}) can be exemplified as follows:
$$\xymatrix@=30pt{[p, q, q]\ar[r]^-{\id_\List\circ\eta}\ar@{=}[rd]&\big[[p], [q], [q]\big]\ar[d]^\mu\\&[p, q, q]}\hspace{.8in}
\xymatrix@=30pt{[p, q, q]\ar[r]^-{\eta\circ\id_\List}\ar@{=}[rd]&\big[[p, q, q]\big]\ar[d]^\mu\\&[p,q,q]}$$
\vspace{.1in}
$$\xymatrix@=30pt{\Big[\big[[p, q], [r]\big], \big[[], [r, q, q]\big]\Big]\ar@{|->}[r]^-{\mu\circ\id_\List}\ar@{|->}[d]_{\id_\List\circ\mu}&\big[[p, q], [r], [], [r, q, q]\big]\ar@{|->}[d]^\mu\\\big[[p, q, r], [r, q, q]\big]\ar@{|->}[r]_\mu&[p, q, r, r, q,q]}$$

\end{example}

The $\List$ monad is but one example of a huge variety of monads on $\Set$. Many more examples will be given in Section \ref{sec:examples}. We now go on to define the Kleisli category associated to a monad. The definition may be a bit opaque. We give an example in \ref{ex:kleisli}, but the real motivation comes in Section \ref{sec:Kleisli instances}. Readers who learn best by example might skip directly to Section \ref{sec:examples}.

\begin{definition}\label{def:kleisli category}

Let $\top=(T,\eta,\mu)$ be a monad on $\Set$. The {\em Kleisli category} associated to $\top$, denoted $\Kls{\top}$, is defined as follows. The objects are sets, i.e. $$\Ob(\Kls{\top})=\Ob(\Set).$$ For any sets $X,Y\in\Ob(\Kls{\top})$ we put $$\Hom_{\Kls{\top}}(X,Y)=\Hom_{\Set}(X,T(Y)).$$ Given morphisms $f\taking X\to Y$ and $g\taking Y\to Z$ in $\Kls{\top}$, we must define their composite $g\circ f$. Unwinding definitions, we are given functions 
\begin{align}X\Too{f}T(&Y)\\
&Y\Too{g} T(Z)
\end{align}
in $\Set$, and we need a function $X\to T(Z)$. Let $\ol{g}\taking T(Y)\to T(Z)$ denote the composite function 
\begin{align}
T(Y)\To{T(g)}T(T(Z))\To{\mu}T(Z).
\end{align}
We define the map $g\circ f\taking X\to Z$ in $\Kls{\top}$ to be the composition of $f$ and $\ol{g}$ in $\Set$. Monad axiom (\ref{dia:monad axiom3}) ensures that this composition law is associative, and axioms (\ref{dia:monad axiom1}) and (\ref{dia:monad axiom2}) ensure that the identity, which on $X\in\Ob(\Kls{\top})$ is given by $\eta_X\taking X\to TX$, is a left and right unit. 

For any set $X$ we refer to elements of $T(X)$ as {\em $T$-values in $X$}.

\end{definition}

\begin{example}\label{ex:kleisli}

We continue working with the $\List$ monad from Example \ref{ex:monad}. The objects of the Kleisli category $\Kls{\List}$ are, as always, simply sets. Given sets $X$ and $Y$ (say $X=\{p,q,r\}$ and $Y=\{s,t\}$), a morphism $f\taking X\to Y$ in $\Kls{\List}$ is a function $X\to\List(Y)$. In other words it consists of three lists in letters $s,t$. For example let us say $$f(p)=[s,s] \hsp\hsp f(q)=[] \hsp\hsp f(r)=[t,s,t].$$ To explain the composition law, let us define a new set $Z=\{u,v\}$ and a function $g\taking Y\to\List(Z)$ given by $$g(s)=[u,v,v]\hsp\hsp g(t)=[v,u].$$ Then the composition $g\circ f\taking X\to Z$ in $\Kls{\List}$ corresponds to the obvious substitution: $$g\circ f(p)=[u,v,v,u,v,v]\hspace{.5in} g\circ f(q)=[]\hspace{.5in} g\circ f(r)=[v,u,u,v,v,v,u].$$

\end{example}

\begin{remark}

Given a monad $\top$, its Kleisli category $\Kls{\top}$ is equivalent to the category of free $\top$-algebras (\cite{BW}). However the database representation of maps that seems to be suggested by the Kleisli category is much more compact than that suggested by the category of free algebras. For example, consider the $\List$ monad. If $X$ is a set with three elements and $Y$ is any set, then a function $X\to \List(Y)$ can be represented by a table with three rows. On the other hand, one might imagine that a map $f\taking\List(X)\to\List(Y)$ should be represented by a table with infinitely many rows, one for each element of $\List(X)$. Our point is that the Kleisli representation is valuable because it is as succinct as possible.

\end{remark}

\section{Kleisli instances}\label{sec:Kleisli instances}

Now that we have a categorical viewpoint of databases (Section \ref{sec:CatDB}) and an understanding of the Kleisli category, we can combine them.

\addtocounter{subsection}{1}

\begin{definition}\label{def:kleisli instances}

Let $\mcC$ denote a schema and let $\top:=(T,\eta,\mu)$ denote a monad on $\Set$ having Kleisli category $\Kls{\top}$. A {\em Kleisli $\top$-instance on $\mcC$} (or simply a {\em $\top$-instance on $\mcC$}) is a functor $\delta\taking\mcC\to\Kls{\top}$. 

\end{definition}

\subsection{Representing Kleisli instances}\label{sec:representing Kleisli}

Let us examine Definition \ref{def:kleisli instances} in detail. For the remainder of Section \ref{sec:representing Kleisli}, $\mcC$ will denote a schema, $\top=(T,\eta,\mu)$ will denote a monad on $\Set$, and $\delta\taking\mcC\to\Kls{\top}$ will denote a Kleisli $\top$-instance on $\mcC$. We will first investigate what information is provided by our instance $\delta$ and then explain how to display it in an extension of the typical database format.

Our schema $\mcC$ consists of objects, arrows, and a path equivalence relation. For each object $c\in\Ob(\mcC)$, our instance provides a set $\delta(c)\in\Ob(\Kls{\top})=\Ob(\Set)$. For each morphism $f\taking c\to c'$ in $\mcC$, our instance provides a morphism $\delta(f)\taking\delta(c)\to\delta(c')$ in $\Kls{\top}$; this is the same as a function $$\delta(f)\taking\delta(c)\to T\delta(c').$$ 

A path $c_0\To{f_1}c_1\To{f_2}c_2\To{f_3}\cdots\To{f_n}c_n$ in $\mcC$ is sent to a composition of functions $$\delta(c_0)\To{\delta(f_1)} T\delta(c_1)\To{\delta(f_2)} T^2\delta(c_2)\To{\delta(f_3)}\cdots\To{\delta(f_n)} T^n\delta(c_n)\To{\mu^{n-1}}T\delta(c_n),$$ and the path equivalence relation must be satisfied with respect to such compositions.

To represent an {\em atomic} database instance $\epsilon\taking\mcC\to\Set$, as in Section \ref{sec:CatDB}, we used (and will continue to use) a tabular format in which every object $c\in\Ob(\mcC)$ was displayed as a table including one ID column and an additional column for every arrow emanating from $c$. In the ID column of table $c$ were elements of the set $\epsilon(c)$ and in the column assigned to some arrow $f\taking c\to c'$ the cells were elements of the set $\epsilon(c')$. 

To represent a {\em Kleisli} database instance $\delta\taking\mcC\to\Kls{\top}$ is similar; we again use a tabular format in which every object $c\in\Ob(\mcC)$ is displayed as a table including one ID column and an additional column for every arrow emanating from $c$. In the ID column of table $c$ are again elements of the set $\delta(c)$; however in the column assigned to some arrow $f\taking c\to c'$ are not elements of $\delta(c')$ but $T$-values in $\delta(c')$, i.e. elements of $T\delta(c')$. 

\begin{example}[Lists]\label{ex:lists}

Let $\top=(\List,\eta,\mu)$ be the list monad as described in Examples \ref{ex:monad} and \ref{ex:kleisli}. Here we show how a $\List$-instance could be represented in a tabular fashion. Our imagined scenario is as follows. We have a set $K$ of tasks. Each task $k\in K$ is composed of an ordered sequence of other tasks in the set, $$\xymatrix{k\ar@{|->}[rrr]^-{\text{is composed of}}&&&[k_1,\ldots,k_n]}.$$ Here a task $k$ might be irreducible ($k\mapsto[k]$) or empty of requirements ($k\mapsto[]$). Our situation is modeled by a $\List$-instance on the schema $$\Loop\iso\fbox{\xymatrix{\LTO{Task}\ar@(l,u)[]^{\text{isComposedOf}}}}$$ The following is an example of such:

$$\begin{tabular}{| l || l |}\bhline
\multicolumn{2}{| c |}{\tt{Task}}\\\bhline 
{\bf ID}&{\bf IsComposedOf}\\\bbhline
a & [b, a, b] \\\hline
b & [e, c]\\\hline
c & [d]\\\hline
d & []\\\hline
e& [d]\\\bhline
\end{tabular}
\hspace{1in}
\parbox{1.5in}
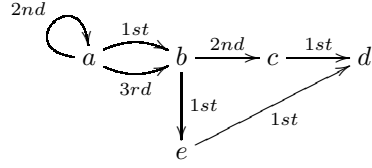
{\xymatrix{
a\ar@(l,u)[]^{2nd}\ar@/^.5pc/[r]^{1st}\ar@/_.5pc/[r]_{3rd}&b\ar[r]^{2nd}\ar[d]^{1st}&c\ar[r]^{1st}&d\\
&e\ar[urr]_{1st}}}
$$

A $\List$-instance on $\Loop$ can be thought of as a directed graph such that every vertex has finitely many outgoing edges, which are linearly ordered.

\end{example}


\subsection{The categories $\mcC\bkls{\top}$ and $\mcC\kls{\top}$}\label{sec:two kleisli instance categories}

Kleisli instances are interesting objects in their own right, as we will see in Section \ref{sec:examples}; however, any category theorist will be interested in the morphisms between them. It seems that different notions of morphisms are appropriate in different circumstances. Below we define two categories for any schema $\mcC$ and monad $\top$; both have the same set of objects, namely the set of Kleisli $\top$-instances on $\mcC$, but one has more morphisms. In Remark \ref{rmk:weird morphisms}, we will offer still another possibility. Perhaps the point is that there several viable notions of morphisms between Kleisli states and the choice of which to use should be dictated by ones purpose.

\begin{definition}\label{def:general}

Let $\mcC$ be a schema, let $\top=(T,\eta,\mu)$ be a monad on $\Set$, and let $\delta,\epsilon\taking\mcC\to\Kls{\top}$ be $\top$-instances on $\mcC$. A {\em general morphism of $\top$-instances on $\mcC$} is a natural transformation $a\taking\delta\to\epsilon$ of functors. We define the {\em category $\top$-instances on $\mcC$}, denoted $\mcC\bkls{\top}$, to be the category having objects and general morphisms as above.

\end{definition}

Let $\mcC$ be a category and $\top=(T,\eta,\mu)$ a monad on $\Set$. Given two $\top$-instances $\delta,\epsilon\taking\mcC\to\Kls{\top}$, a general morphism $a\taking\delta\to\epsilon$ is simply a natural transformation of functors. Unpacking that definition, we have for every object $c\in\Ob(\mcC)$ a component morphism $a_c\taking\delta(c)\to\epsilon(c)$ in $\Kls{\top}$, which is a function $a_c\taking\delta(c)\to T\epsilon(c)$. These components have to fit into naturality squares: given any $f\taking c\to c'$ in $\mcC$ we need the following diagram to commute: 
\begin{align}\label{dia:natural kleisli}
\xymatrix{\delta(c)\ar[r]^{\delta(f)}\ar[d]_{a_c}&T\delta(c')\ar[d]^{\mu\circ Ta_{c'}}\\T\epsilon(c)\ar[r]_{\mu\circ T\epsilon(f)}&T\epsilon(c').}
\end{align}

We will see in Example \ref{ex:rep theory} that for the classical mathematical subject area of representation theory \cite{EGH}, these so-called general morphisms are precisely what one wants. In other words for the vector-space monad $\Vect$, and a category $G$ (generally either a group or a quiver), the category of general $\Vect$-instances is the category of $G$-representations, $G\bkls{\Vect}\iso\Rep(G)$. 

However for classical computer science, these general morphisms seem to be too general. For example, we will show that there is a monad $\Tur^U_M$ for which the $\Loop$-instances are almost precisely the same thing as Turing machines.\footnote{To be explicit, an object in $\Loop\kls{\Tur^U_M}$ is equivalent to a Turing machine for which the start state has not been specified. We call objects in $\Loop\kls{\Tur^U_M}$ {\em unpointed Turing Machines}.} In this setting, general morphisms seem strange and unmotivated whereas the basic morphisms (Definition \ref{def:basic}) make much more sense. Indeed, given a basic morphism $p\taking\delta\to\delta'$ of unpointed Turing machines and any choice of start state $S\in\delta(s)$, the Turing machine $(\delta,S)$ computes the same function as $(\delta',p(S))$ computes.

\begin{definition}\label{def:basic}

Let $\mcC$ be a schema, let $\top=(T,\eta,\mu)$ be a monad on $\Set$, and let $\delta,\epsilon\taking\mcC\to\Kls{\top}$ be $\top$-instances on $\mcC$. A {\em basic morphism} $b\taking\delta\to\epsilon$ consists of a component function $b_c\taking\delta(c)\to\epsilon(c)$ for each object $c\in\Ob(\mcC)$, such that for each morphism $f\taking c\to c'$ in $\mcC$ the induced diagram of sets commutes:
\begin{align}\label{dia:basic morphism}
\xymatrix{\delta(c)\ar[d]_{b_c}\ar[r]^{\delta(f)}&T\delta(c')\ar[d]^{Tb_{c'}}\\\epsilon(c)\ar[r]_{\epsilon(f)}&T\epsilon(c').}
\end{align}
We denote by $\mcC\kls{\top}\ss\mcC\bkls{\top}$ the {\em basic subcategory of $\top$-instances on $\mcC$} which has as objects all $\top$-instances and which has as morphisms the basic morphisms, as defined above.

\end{definition}

\begin{remark}\label{rmk:basic}

Given $\delta,\epsilon\taking\mcC\to\Kls{\top}$, there is another definition of basicness that is equivalent and perhaps more categorical. Namely a natural transformation $a\taking\delta\to\epsilon$ is basic if and only if, for each object $c\in\Ob(\mcC)$ the component $a_c\taking\delta(c)\to T\epsilon(c)$ factors through the unit component $\eta_{\epsilon(c)}\taking\epsilon(c)\to T\epsilon(c).$

\end{remark}

For any monad $\top=(T,\eta,\mu)$ on $\Set$ and any category $\mcC$, there is a functor from the category of ordinary (atomic) database instances into the basic category of $\top$-instances, $$E^\mcC_\top\taking\mcC\set\to\mcC\kls{\top}.$$ For an object $\delta\taking\mcC\to\Set$ we have $E^\mcC_\top(\delta)=\delta$, and for a morphism $a\taking\delta\to\epsilon$, we have $E^\mcC_\top(a)=\eta_\epsilon\circ a\taking\delta\to T\epsilon$. This functor is a faithful if and only if there exists a set $X$ such that $\top(X)$ has cardinality at least 2. 

\comment{

There is also an inclusion of the category of basic instances into the general category, $$L_\top^\mcC\taking\mcC\kls{\top}\to\mcC\bkls{\top}.$$ The lemma below states that the basic category is almost always different from the general category of Kleisli instances.

\begin{lemma}\label{lemma:basic is different than general}

Suppose that $\mcC$ is a nonempty category and $\top=(T,\eta,\mu)$ a monad on $\Set$. Then $L_\top^\mcC\taking\mcC\kls{\top}\to\mcC\bkls{\top}$ is an equivalence of categories if and only if $\eta\taking\id_\Set\to T$ is an isomorphism of functors $\Set\to\Set$.

\end{lemma}

\begin{proof}

If $\eta$ is an isomorphism then $L_\top^\mcC$ is an equivalence by Remark \ref{rmk:basic}. For the other direction, suppose that $L^\mcC_\top$ is an equivalence of categories and let $c\in\Ob(\mcC)$ be an object considered as a functor $[0]\To{c}\mcC$. Then composition with $c$ yields an isomorphism $\mcC\kls{\top}\iso\Set$ and $\mcC\bkls{\top}\iso\Kls{\top}$. Thus we have a commutative square 
$$\xymatrix{\mcC\kls{\top}\ar[r]^{L^\mcC_\top}\ar[d]_\iso&\mcC\bkls{\top}\ar[d]^\iso\\\Set\ar[r]&\Kls{\top}.}$$
The top functor $L^\mcC_\top$ is assumed to be an equivalence, so the bottom functor is too, and this implies that $\eta$ is a natural isomorphism.

\end{proof}

}
\begin{remark}\label{rmk:weird morphisms}

Another notion of morphisms between $\top$-instances on $\mcC$ is often appropriate. If $\top$ is the extension of a monad $\top'=(T',\eta',\mu')$ on $\Cat$ via the adjunction $\Adjoint{\text{Disc}}{\Set}{\Cat}{\Ob}$, then for any $X\in\Ob(\Set)$ one may say that ``$T(X)$ naturally has the structure of a category" because $T'(\text{Disc}(X))$ is a category and $T(X)=\Ob(T'(X))$. The monads $\mcP$, Multiset, List, and $\Exc_E$ (from Sections \ref{sec:multiset},\ref{sec:list}, and \ref{sec:exceptions}) are instances of this phenomenon. Consider, for example, the monad $\mcP$ where one recognizes that the power set of any set $X$ does naturally come with a partial order. In this case, the monad $\mcP$ on $\Set$ is induced by a monad on $\Cat$ whose functor part is $\mcX\mapsto\Fun(\mcX,[1])$, where $[1]=\fbox{$\LMO{}\too\LMO{}$}$ is the walking arrow category.

When $\top$ extends to a monad on $\Cat$, there seems to be another natural notion of morphisms between $\top$-instances on a schema $\mcC$. Namely, a {\em lax morphism} $a\taking\delta\to\epsilon$ between $\delta,\epsilon\taking\mcC\to\Kls{\top}$ consists of a component function $a_c\taking\delta(c)\to\epsilon(c)$ for each $c\in\Ob(\mcC)$ and, for each $f\taking c\to c'$ in $\mcC$, a natural transformation diagram 
$$
\xymatrix{\delta(c)\ar[r]^{\delta(f)}\ar[d]_{a_c}\ar@{}[dr]|{\Swarrow}&T(\delta(c'))\ar[d]^{a_{c'}}\\
\epsilon(c)\ar[r]_{\epsilon(f)}&T(\epsilon(c'))}
$$ 
in other words, a morphism $a_{c'}\circ\delta(f)\too\epsilon(f)\circ a_c$.

When it is defined, this notion of morphism seems to have some advantages. For example when $T=$Multiset,  we will see in Example \ref{ex:graphs} that the objects of $\Loop\kls{\text{Multiset}}$ are graphs, but the morphisms are more restrictive than graph morphisms. On the other hand, the category with the same objects and lax morphisms is equivalent to the category of graphs.

\end{remark}



\section{Examples}\label{sec:examples}

In this section we provide a survey of available monads that may be useful in databases. We divide them into five roughly sensible groups. In Section \ref{sec:universals} we discuss two monads, one of which is initial in the category of monads (and gives rise to ordinary (atomic) database instances) and one of which is terminal in the category of monads (and gives rise to so-called unlinked instances). In Section \ref{sec:collections} we give examples of monads that represent various kinds of collection such as subsets, multisets, lists, and probability distributions. In Section \ref{sec:tunable} we discuss monads which we describe as ``tunable," because one can adjust the choice of monad in a controlled way; for example, for each choice of set $E$ one obtains a different monad $\Exc_E$ of $E$-exceptions. In Section \ref{sec:algebraic} we consider various classical algebraic monads, e.g. of vector spaces. Finally in Section \ref{sec:processes} we consider process-oriented monads that include computations and experiments.

\subsection{Universal monads}\label{sec:universals}

\subsub{Atomic instances}\label{sec:atomic}

There is an identity monad on $\Set$, which we denote $$\text{Atomic}=(\id_\Set,\id,\id).$$ Its instances are here called {\em atomic} (or {\em ordinary}, or {\em ordinary atomic}) instances. See Examples \ref{ex:atomic worksIn} and \ref{ex:loop}.

\subsub{Unlinked entities}\label{sec:unlinked}

Consider the monad $$\text{Unlinked}=(\singleton^-,!,!),$$ where for any set $X\in\Set$, the set $\singleton^X:=\singleton$ is the terminal object in $\Set$. The unit and multiplication maps are completely determined by their domain and codomain. An Unlinked-instance on $\mcC$ includes a set of records for each object $c\in\Ob(\mcC)$, but the foreign keys offer no connection between them.

\begin{example}

Let $\mcW$ be as in Example \ref{ex:atomic worksIn}. An instance $\delta\taking\mcW\to\Kls{\text{Unlinked}}$ might look like
$$\delta:=\left\{\hsp
\tiny\begin{tabular}{| l || l |}\bhline
\multicolumn{2}{| c |}{\tt{Employee}}\\\bhline 
{\bf ID}&{\bf WorksIn}\\\bbhline
Alice & *\\\hline
Bob & *\\\hline
Carl & *\\\hline
Deb & *\\\hline
Fred & *\\\hline
Jen & *\\\bhline
\end{tabular}
\hspace{.5in}
\begin{tabular}{| l ||}\bhline
\multicolumn{1}{| c |}{\tt{Department}}\\\bhline 
{\bf ID}\\\bbhline
Bio\\\hline
EECS\\\hline
Math\\\hline
Music\\\bhline
\end{tabular}
\hsp\right\}
$$

\end{example}

\subsection{Collection monads}\label{sec:collections}

\subsub{Subsets}\label{sec:subsets}

The monad $$\mcP=(\PP,\{-\},\cup)$$ sends a set to its power set; the unit $\{-\}\taking X\to \PP(X)$ is given by singleton subsets, $x\mapsto\{x\}$, and the multiplication is given by union $\cup\taking\PP(\PP(X))\to\PP(X)$. Note that there is an isomorphism of categories $\Kls{\mcP}\iso\Rel$, where $\Rel$ is the category of sets and binary relations \cite{FS}. 

\begin{example}

$$\delta:=\left\{\hsp
\small\begin{tabular}{| l || l |}\bhline
\multicolumn{2}{| c |}{\tt{Employee}}\\\bhline 
{\bf ID}&{\bf WorksIn}\\\bbhline
Alice & \{Math, EECS\}\\\hline
Bob & \{EECS\}\\\hline
Carl & \{\}\\\hline
Deb & \{EECS\}\\\hline
Fred & \{Math, Bio\}\\\hline
Jen & \{Bio\}\\\bhline
\end{tabular}
\hspace{.5in}
\begin{tabular}{| l ||}\bhline
\multicolumn{1}{| c |}{\tt{Department}}\\\bhline 
{\bf ID}\\\bbhline
Bio\\\hline
EECS\\\hline
Math\\\hline
Music\\\bhline
\end{tabular}
\hsp\right\}
$$

\end{example}

\begin{example}[Nonempty subsets]\label{ex:nonempty subsets}

It is easy to see that the nonempty subsets functor $\PP_+$ given by $\PP_+(X)=\{Y\ss X\|Y\neq\emptyset\}$ is the functor part of a monad. Note that $\Kls{\PP_+}$ is the set of correspondences in the sense of theoretical economics, e.g. for ``best response" strategies in game theory \cite[Section 2.1.5]{Car}.

\end{example}

\begin{example}[Turning a database inside out]

Given a category $\mcC$ and an ordinary database instance $\delta\taking\mcC\to\Set$ we can, in a sense, invert $\delta$ by producing an instance on $\mcC\op$: $$\back{\delta}\taking\mcC\op\to\Kls{\mcP}.$$ For $c\in\Ob(\mcC)$ we have $\back{\delta}(c)=\delta(c)$. For $f\taking c'\to c$ in $\mcC$ and $x\in\delta(c)$, we define $\back{\delta}(f)\taking\delta(c)\to\PP(\delta(c'))$ by $$\back{\delta}(f)(x)= \delta(f)^\m1(x)\ss\delta(c').$$

\end{example}

\subsub{Lists}\label{sec:list}

This was the running example in Section \ref{sec:monads and Kleisli} and Section \ref{sec:Kleisli instances}; see in particular Example \ref{ex:lists}.

One can also define a {\em non-empty lists} monad $\List_+$, similarly to the nonempty subsets monad $\PP_+$ of Example \ref{ex:nonempty subsets}.

\subsub{Finite multisets}\label{sec:multiset}

The monad $\text{Multiset}=(T,\eta,\mu)$ is given as follows. The functor part $T\taking\Set\to\Set$ is given by $$T(X):=\coprod_{n\in\NN}(X^n/\Sigma_n)$$ where $\Sigma_n$ is the symmetric group on $n$ letters, which acts on elements $x\in X^n$ by permuting the order of entries in $x$; the quotient of this action is the set of un-ordered $n$-tuples in $X$. The unit and multiplication in the Multiset monad are analogous to the unit and multiplication in the List monad.

The category $\Kls{\text{Multiset}}$ is equivalent to the category of sets and correspondences. Recall \cite[Section 2.3.1]{Lur} that for sets $X$ and $Y$, a correspondence between $X$ and $Y$ is a diagram of the form $X\From{f}C\To{g}Y$, where $C$ is a set and $f,g$ are functions. 

\begin{example}\label{ex:graphs}

Given a set $X$, a correspondence from $X$ to itself is a graph with vertex set $X$. In other words, a graph is precisely a Multiset-instances on the category $\Loop$. We denote a multiset using usual set notation, but in which duplicate entries with the same name correspond to distinct elements of the multiset.

\begin{align}
\delta:=
\begin{tabular}{| l || l |}\bhline
\multicolumn{2}{| c |}{\tt{s}}\\\bhline 
{\bf ID}&{\bf f}\\\bbhline
a & \{b,d\}\\\hline
b & \{c,c\}\\\hline
c & \{\}\\\hline
d & \{\}\\\hline
e & \{e\}\\\bhline
\end{tabular}
\hspace{1in}
\parbox{.1in}{
\xymatrix{\LMO{a}\ar[r]\ar[d]&\LMO{b}\ar@/^1pc/[r]\ar@/_1pc/[r]&\LMO{c}\\\LMO{d}&\LMO{e}\ar@(u,r)[]}
}
\end{align}

The usual notion of graph morphism is captured by the lax notion given in Remark \ref{rmk:weird morphisms}.

\end{example}

\subsub{Distributions}\label{sec:distributions}

Let $[0,1]\ss\RR$ denote the set of real numbers between $0$ and $1$. Let $X$ be a set and $p\taking X\to[0,1]$ a function. We say that $p$ is a {\em finitary probability distribution on $X$} if there exists a finite subset $W\ss X$ such that 
\begin{align}\label{dia:sum to 1}
\sum_{w\in W}p(w)=1,
\end{align} and such that for all $x'\in X-W$ in the complement of $W$ we have $p(x')=0.$ Note that $W$ is unique if it exists; we call it {\em the support of $p$} and denote it $\Supp(p)$. Note also that if $X$ is a finite set then every function $p$ satisfying (\ref{dia:sum to 1}) is a finitary probability distribution on $X$.

For any set $X$, let $\Dist(X)$ denote the set of finitary probability distributions on $X$. It is easy to check that given a function $f\taking X\to Y$ one obtains a function $\Dist(f)\taking\Dist(X)\to\Dist(Y)$ by $\Dist(f)(y)=\sum_{f(x)=y}p(x)$. Thus we can consider $\Dist\taking\Set\to\Set$ as a functor, and in fact the functor part of a monad. Its unit $\eta\taking X\to\Dist(X)$ is given by the Kronecker delta function $x\mapsto \delta_x$ where $\delta_x(x)=1$ and $\delta_x(x')=0$ for $x'\neq x$. Its multiplication $\mu\taking\Dist(\Dist(X))\to\Dist(X)$ is given by weighted sum: given a finitary probability distribution $w\taking\Dist(X)\to[0,1]$ and $x\in X$, put $\mu(w)(x)=\sum_{p\in\Supp(w)}w(p)p(x).$ 

\begin{example}[Markov chains]\label{ex:markov}

Let $\Loop$ be as in Example \ref{ex:loop}. A $\Dist$-instance on $\Loop$ is equivalent to a time-homogeneous Markov chain. To be explicit, a functor $\delta\taking\Loop\to\Kls{\Dist}$ assigns to the unique object $s\in\Ob(\Loop)$ a set $S=\delta(s)$, which we call the state space, and to $f\taking s\to s$ a function $\delta(f)\taking S\to\Dist(S)$, which sends each element $x\in S$ to some probability distribution on elements of $S$. For example, the table $\delta$ on the left corresponds to the Markov matrix $M$ on the right below:
\begin{align}
\delta:=
\begin{tabular}{| l || l |}\bhline
\multicolumn{2}{| c |}{\tt{s}}\\\bhline 
{\bf ID}&{\bf f}\\\bbhline
1 & .5(1)+.5(2)\\\hline
2 & 1(2)\\\hline
3 & .7(1)+.3(3)\\\hline
4 & .4(1)+.3(2)+.3(4)\\\bhline
\end{tabular}
\hspace{.5in}
M:=\left(
\begin{array}{cccc}
0.5 & 0.5 & 0 & 0\\
0 & 1 & 0 & 0\\
0.7 & 0 & 0.3 & 0\\
0.4 & 0.3 & 0 &0.3
\end{array}
\right)
\end{align}

As one might hope, for any natural number $n\in\NN$ the map $f^n\taking S\to\Dist(S)$ corresponds to the matrix $M^n$, which sends an element in $S$ to its probable location after $n$ iterations of the transition map.

One could also at least encode the information necessary to describe time-inhomogeneous Markov chains by using similar schemas, such as $$\Loop\vee\Loop:=\fbox
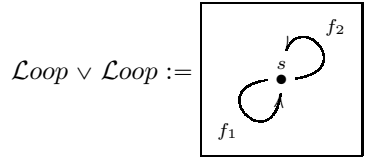
{\xymatrix{\LMO{s}\ar@(r,u)[]_{f_2}\ar@(l,d)[]_{f_1}}}$$ or $\Loop^{\vee\{1,2,\ldots\}}=\Loop\vee\Loop\vee\cdots$, and the same monad $\Dist$.

\end{example}

\subsection{Tunable monads}\label{sec:tunable}

Some monads on $\Set$ come in families. To make this precise, we will say that a {\em tunable monad} is a pair ($\mcI,P)$ where $\mcI$ is a small category and $P\taking\mcI\to\Monad_\Set$ is a functor (see Definition \ref{def:category of monads}). Of course then any monad can be trivially considered tunable by taking the indexing category to be $\mcI=\fbox{$\bullet$}$ (the discrete category on one object). It is clear that the degree to which a monad is tunable is measured by the complexity of $(\mcI,P)$. In the present section we will only mention three such monads; the first is indexed by $\Fin$, the category of finite sets, the second is indexed by $\Fin\op$, and the third is indexed by the category of monoids. See Section \ref{sec:transformations} for more on the value of tunable monads.

\subsub{Exceptions}\label{sec:exceptions}

Let $E\in\Fin$ be a finite set. The monad $\Exc_E=(-\amalg E,\eta,\mu)$ is given as follows. The unit $\eta_X\taking X\to X\amalg E$ is given simply by the inclusion. The multiplication $\mu_X\taking X\amalg E\amalg E\to X\amalg E$ is given in the obvious way, by identity on each copy of $X$ and $E$.

\begin{example}\label{ex:atomic and maybe}

If $E=\emptyset$ the exception monad reduces to the identity monad, i.e. $\Exc_{\emptyset}$-instances are ordinary atomic instances. If $E=\singleton$ is the one-element set, then $\Exc_\singleton$-instances correspond to databases in which a field can have null values; we call $\Exc_\singleton$ the {\em Maybe monad}.

\end{example}

\begin{example}

Let $\Loop:=\LoopSchema$. Then for any set $E$, the $\Exc_E$-instances on $\Loop$ can encode recursive functions with output values in $E$. For example, with $E=\NN$ we obtain the factorial function $n\mapsto n!$ as an $\Exc_E$-instance, $\delta\taking\Loop\to\Kls{\Exc_\NN}$. Namely, we put $\delta(s):=\NN\cross\NN$ and we put $\delta(f)\taking\delta(s)\to\delta(s)\amalg E$ on $(m,n)\in\delta(s)$ by 
$$\delta(m,n):=
\begin{cases}
(mn,n-1)\in\delta(s)& \text{ if } n\geq 1\\
m\in E& \text{ if } n=0.
\end{cases}
$$
Then for any $n\in\NN$, the factorial of $n$ is obtained by starting with $(1,n)$ and repeatedly applying $\delta(f)$ until an output (in $E=\NN$) is returned. 

\end{example}

\begin{example}[Database schemas]

In Section \ref{sec:CatDB} we gave a definition of database schemas, but we did not mention data types. One model for typed database schemas can be found in \cite[Section 5.1]{Sp1}, but here we present another model based on monads. 

Let $E$ be a set, the elements of which are names of datatypes, e.g. $E=\{\text{String, Int, Float}\}$. Let $\List_E=(T_E,\eta,\mu)$ be the monad with functor part $T_E(X)=\List(X\amalg E)$, sending a set $X$ to the set of lists for which each entry is an element either of $X$ or of $E$. Then we construe any instance $\delta\in\Loop\bkls{\List_E}$ as a database schema in the following way. The set $\delta(s)$ serves as the set of tables, and for each table $x\in\delta(s)$ the list $\delta(f)(x)$ serves as the set of columns of $x$, each of which is either another table (indicating a foreign key) or a datatype. 

\end{example}

\subsub{Inputs}\label{sec:input}

Let $U\in\Fin$ be a finite set. The monad $\Inp^U=(X\mapsto X^U,\eta,\mu)$ is given as follows. The unit $\eta_X\taking X\to X^U$ sends $x$ to the constant function at $x$. If $\Delta_U\taking U\to U\cross U$ is the diagonal map, then we can describe the multiplication $\mu_X\taking (X^U)^U\to X^U$ by $$(X^U)^U\iso X^{U\cross U}\To{\;\Delta_U}X^U.$$

\begin{example}[Tailored user experience]

If $U$ is a set of users then the database instance $\delta\taking\mcC\to\Kls{\Inp^U}$ would provide possibly different values for different $u\in U$. 

Similarly, if $U$ is the set of dates, then the values in a database instance $\delta$ could be made to depend on the date.

\end{example}

\begin{example}[Each universal monad as a special case]

With $U=\emptyset$ the input monad $\Inp^\emptyset$ reduces to the Unlinked monad of Section \ref{sec:unlinked}. If $U=\singleton$, the input monad $\Inp^{\singleton}$ reduces to the identity monad, whose Kleisli instances are ordinary atomic instances as in Section \ref{sec:atomic}.

\end{example}

\begin{example}[Finite state automata]\label{ex:fsa}

A finite state automaton consists of a set $S$ of states, a set $T$ of transitions, and a function $T\cross S\to S$. By currying, this can be rewritten as a function $S\to S^T$. The category of finite state automata with transitions $T$ is precisely the category $\Loop\kls{\Inp^T}$ of $T$-Input instances on $\Loop$.

\end{example}

\subsub{Monoid annotation}\label{sec:monoid}

Let $(M,1,\star)$ be a monoid. We define the monad $$M\text{-annotated}:=(X\mapsto M\cross X, (1,-), (\star,-)).$$ One way to think about this is that $M$ is a language of instructions, multiplication corresponding to carrying out a sequence of instructions and unit corresponding to doing nothing, and $M$-annotated instances keep track of such instructions as one follows foreign keys through the database. However, there are other ways to think about $M$-annotated instances as well, as we show in two examples.

\begin{example}[Assurance]\label{ex:assurance}

Consider the monoid $M=([0,1],1,\ast)$, where $[0,1]\ss\RR$ is the unit interval, and $\ast\taking[0,1]\cross[0,1]\to[0,1]$ is given by multiplication of real numbers. Think of a $[0,1]$-annotated values as assurances. In other words if a data entry clerk or a scientist is less than 100\% sure that a certain datum is correct, they can annotate it with their assurance level. To keep things uncluttered we simply do not write our assurance value if it is unit (100\%).

$$\delta:=
\small\begin{tabular}{| l || l |}\bhline
\multicolumn{2}{| c |}{\tt{Person}}\\\bhline 
{\bf ID}&{\bf LivesAt}\\\bbhline
Alice & 15 Ashville Rd.\\\hline
Bob & 34 Vine St. (80\%)\\\hline
Carl & 21 Post St. (90\%)\\\hline
Deb & 110 W. 5th Ave.\\\bhline
\end{tabular}
$$

The monad multiplication assures that probability values will propagate through the database (with an independence assumption) as we compose foreign keys.

\end{example}

\begin{example}[Time-delay]\label{ex:time delay}

Consider the monoid $M=(\RR_{\geq 0},0,+)$, where $\RR_{\geq 0}\ss\RR$ is the set of non-negative real numbers. Think of $\RR_{\geq 0}$-annotated values as time-delays. In other words, each foreign key $f\taking c\to d$ in a database may correspond to a process that converts things of type $c$ into things of type $d$, and the time delay monad allows us to also encode how long that process is expected to take. Monad multiplication assures that these values will be added together as we string together longer processes by composing foreign keys. Note that we could use $\RR_{\geq 0}\cup\{\infty\}$ instead of $\RR_{\geq 0}$ if we wanted to allow for never-ending processes.

\end{example}

\subsub{Turing Machines}\label{sec:turing machines}

Let $M$ be a monoid and $U$ a finite set. We have seen that on $\Loop$, the monad $\Inp^U$ encodes finite state automata (Example \ref{ex:fsa}), the monad $M\text{-annotated}$ encodes finite lists of instructions (Section \ref{sec:monoid}), and the monad $\Exc_E$ encodes exceptions (Section \ref{sec:exceptions}). Let us set $E=\{\Halt\}$. We can combine these three monads into a new monad: $$\Tur_M^U=\Big(X\mapsto \big(M\cross (X\amalg\Halt)\big)^U,\mu,\eta\Big)$$ We do not describe the unit and multiplication here, but they are easy enough to reconstruct in analogy with the descriptions in Sections \ref{sec:exceptions}, \ref{sec:input}, and \ref{sec:monoid}, assuming $M$ acts trivially on \{Halt\}.

\begin{example}\label{ex:turing}

Consider the case where $U=\{0,1\}$ and where $M$ is the free monoid on the set $\{L,R,W_0,W_1\}$, which we think of as the set of all sequences of instructions to move left, move right, write a 0, and write a 1. Then if $X$ is thought of as the set of states of a Turing machine, a function $U\to M\cross (X\amalg\Halt)$ reads the input and produces an instruction and a new state (possibly the Halt state). 

$$\xymatrix{\fbox{Start}\ar@<.5ex>[r]^{0:W_1}\ar@<-.5ex>[r]_{1:W_1}&\fbox{q0}\ar@(r,u)[]_{1:L}\ar[r]_{0:W_1}&\fbox{q1}\ar@/_2pc/[rr]_{0:W_0}\ar[r]^{1:L}&\fbox{q2}\ar@(r,u)[]_{1:W_0}\ar[r]_{0:R}&\fbox{Halt}\ar@(r,u)[]_{1:W_1}\ar@(r,d)[]^{0:W_0}}
$$
\\
$$
\Loop:=\LoopSchema\hspace{.8in}\small
\delta:=\begin{tabular}{| l || l |}\bhline
\multicolumn{2}{| c |}{\tt{s}}\\\bhline 
{\bf ID}&{\bf f}\\\bbhline
Start & 0: ($W_1$, q0), 1: ($W_1$, q0)\\\hline
q0 & 0: ($W_1$, q1), 1: ($L$, q0)\\\hline
q1 & 0: ($W_0$, Halt), 1: ($L$, q2)\\\hline
q2 & 0: ($R$, Halt), 1: ($W_0$, q2)\\\bhline
\end{tabular}
$$

A functor $\delta\taking\Loop\to\Kls{\Tur^U_M}$ consists of a set $X=\delta(s)$ of states and a function $\delta(f)\taking X\to (M\cross (X\amalg\Halt)^U$, which can be curried to $X\cross U\to M\cross (X\amalg\Halt)$. After we choose a start state, we find ourselves with precisely the specification of Turing machines given in \cite{BJ}.

\end{example}

Tangentially, one may wonder how to evaluate such a Turing machine. Let $\Tape$ denote the set of positioned tapes, i.e. pairs $(T,p)$ where $T\taking\ZZ\to\{0,1\}$ is a function and $p\in\ZZ$. There is an evaluation function $e\taking\Tape\to U$ given by $e(T,p):=T(p)$. By construction we have an action $\alpha\taking M\cross\Tape\to\Tape$. We have a natural transformation $E\taking\Tur^U_M(-)\to(\Tape\cross(-\amalg\Halt))^\Tape$, given on $X$ by 
\begin{align}
\nonumber (M\cross (X\amalg\Halt))^U&\Too{e} (M\cross (X\amalg\Halt))^{\Tape}\\\label{dia:turing spec to impl}
&\To{\id_\Tape}(M\cross\Tape\cross (X\amalg\Halt))^{\Tape}\\
\nonumber&\Too{\alpha}(\Tape\cross (X\amalg\Halt))^{\Tape}.
\end{align}

Choose a turing machine $\delta\taking\Loop\to\Kls{\Tur^U_M}$ with start state $S\in X$ and let $I\in\Tape$ be the initialized tape. Then for each $n\in\NN$ we have $\big(E\circ\delta(f^n)\big)(S)(I)\in\Tape\cross(X\amalg\Halt)$, and we proceed with increasing values of $n$ until the function returns the halt state, at which point we output the tape. 

In Section \ref{sec:morphisms of monads} we will discuss morphisms of monads. We caution the reader that while $\Tur^U_M$ and $X\mapsto (\Tape\cross (X\amalg\Halt))^\Tape$ are both monads, the mapping $E$ in (\ref{dia:turing spec to impl}) is {\em not} a morphism of monads. It is a natural transformation of functors that preserves the unit but not the multiplication. This failure is somehow expected: if there were a morphism of monads from a Turing machine's specification to its implementation, the behavior of programs would be more easily analyzed than it turns out to be.

\subsection{Algebraic monads}\label{sec:algebraic}

\subsub{Vector spaces}

Let $k$ be a field. There is a $k$-vector-space monad sending a set $X$ to the free $k$-vector space with basis $X$. The unit map corresponds to the inclusion of basis vectors and the multiplication map corresponds to the ability to convert a linear combination of vectors into a single vector.

\comment{
\begin{example}[Euler's method]

Let $\Loop$ be the loop category shown in Example \ref{ex:loop}. $\Vect$-instances on $\Loop$ can be understood as encoding Euler's method calculations, as applied to homogeneous ordinary differential equations. For example applying Euler's method with $\Delta t=0.1$ to the differential equation on the left yields the equation on the right, which is encoded in the $\Loop\kls{\Vect}$ table below:
$$\left(
\begin{array}{c}
\frac{d x}{d t}\\\\\frac{d y}{d t}
\end{array}
\right)
=
\left(
\begin{array}{cc}
3x&4y\\\\-2x&y
\end{array}
\right)
\hspace{.5in}
\left(
\begin{array}{c}
x_{n+1}\\y_{n+1}
\end{array}
\right)
=
\left(
\begin{array}{cc}
1.3x_n&0.4y_n\\-0.2x_n&1.1y_n
\end{array}
\right)
$$ \\
$$
\delta:=\begin{tabular}{| l || l |}\bhline
\multicolumn{2}{| c |}{\tt{s}}\\\bhline 
{\bf ID}&{\bf f}\\\bbhline
x & 1.3x + 0.4y\\\hline
y & -0.2x + 1.1y\\\bhline
\end{tabular}
$$

After $n$ steps, Euler's method outputs the same values as the composite $\delta(f^n)$ does.

\end{example}

}

\begin{example}[Representation theory]\label{ex:rep theory}

If $G$ is a group (considered as a category with one object) then $G\bkls{\Vect_k}$ is equivalent to $\Rep_k(G)$, the category of $G$-representations. If $Q$ is a free category then $Q\bkls{\Vect_k}$ is equivalent to $\Rep_k(Q)$, the category of quiver representations on $Q$ (see \cite{Kac}). In particular, Jordan Canonical Form is the classification of isomorphism classes in $\Loop\bkls{\Vect_\CC}$.

\end{example}

\subsub{Others}

There are many algebraic theories---monoids, commutative monoids, groups, abelian groups, rings, commutative rings, etc., to name a few. In fact, some authors \cite{Le2} {\em define} algebraic theories simply as monads on $\Set$. Each monad on $\Set$ has an associated Kleisli category. In fact, two such monads have already been mentioned above under different names. The $\List$-monad from Example \ref{ex:lists} is another name for the monoid monad, and the Multiset monad from Section \ref{sec:multiset} is another name for the commutative monoid monad. 

\begin{example}[Multigraphs]\label{ex:multigraphs}

A {\em multigraph} (see \cite{HMP}) consists of a set of nodes and a set of {\em multi-arrows}, each of which points from one node to a finite list of nodes. A symmetric multigraph is almost the same except each multi-arrow points from one node to a finite set of nodes. 

Let $\NN[-]\taking\Set\to\Set$ denote the (functor part of the) free commutative rig monad \cite{Gol} (respectively let $\NN\langle-\rangle\taking\Set\to\Set$ denote the free rig monad). For example $\NN[x,y]$ is the set of polynomials in $x,y$ with natural number coefficients, containing elements like $3$ and $xy^2+2x^3+y$. (Similarly $\NN\langle x,y\rangle$ would contain elements like $xyy+yxy$.) The set of $\Kls{\NN[-]}$-instances (resp. the set of $\Kls{\NN\langle-\rangle}$-instances) on $\Loop$ is precisely the set of symmetric multigraphs (resp. multigraphs). With morphisms as in Remark \ref{rmk:weird morphisms}, the category of $\NN\langle-\rangle$-instances on $\Loop$ is equivalent to the category of multigraphs.


\end{example}

\subsection{Process monads}\label{sec:processes}

The examples in this section are a bit more far-flung, but still may be useful to give an idea of what is possible.

\begin{example}[Computation]\label{ex:computation}

Fix a programming language $L$. For any set $X$, let $T(X)$ denote the set of programs that are written in $L$ and such that, taking no input, will either halt and return a value in $X$ or not halt. This is the functor part of a monad. By currying, a Kleisli map $X\to T(Y)$ is equivalent to a (possibly non-halting) computation taking input in $X$ and returning values in $Y$.

\end{example}

\begin{example}[Experiment]\label{ex:experiment}

For any set $X$, let $T(X)$ denote the set of specifications for experiments that could be carried out from a fixed initial condition and that will result in a value in $X$. For example, if $X=\ZZ$ is the set of integers, then $T(X)$ might include as an element the phrase ``survey 100 customers at the McDonalds on 32nd street, asking their favorite real number. Take their average as a real number and then apply the floor function to obtain an integer". Then $T$ can be construed as the functor part of a monad. A Kleisli map $X\to TY$ is equivalent to the specification of an experiment that takes parameters in $X$ and outputs a value in $Y$. The point is that the database does not hold the {\em results} of these experiments, but instead the {\em experiment specifications} which, if performed, will result in values later. Any value counts as a (trivial) experiment, so this generalizes ordinary databases.

\end{example}

\section{Transformations}\label{sec:transformations}

Monads, like everything in category theory, are not stand-alone objects but exist in a category, in which the morphisms are an integral part of the picture. In section \ref{sec:morphisms of monads} we will define morphisms of monads. These include operations like transforming a list into a multiset (by forgetting order) or transforming a probability distribution into a subset (by taking all elements that have nonzero probability). A morphism $f\taking \top\to\top'$ of monads results in a functor between the corresponding Kleisli categories. For any schema $\mcC$ we have a commutative square 
$$
\xymatrix{\mcC\kls{\top}\ar[r]^f\ar@{^(->}[d]&\mcC\kls{\top'}\ar@{^(->}[d]\\
\mcC\bkls{\top}\ar[r]_f&\mcC\bkls{\top'}}
$$ 
that converts $\top$-instances into $\top'$-instances in either sense given below (see Definitions \ref{def:general} and \ref{def:basic}). In Section \ref{sec:examples of transformations} we will sketch some examples. 

\subsection{Morphisms of monads}\label{sec:morphisms of monads}

\begin{definition}\label{def:morphisms of monads}

Let $\top=(T,\eta,\mu)$ and $\top'=(T',\eta',\mu')$ be monads on $\mcS$. A {\em morphism of monads from $\top$ to $\top'$} is a natural transformation $\alpha\taking T\to U$ such that the following diagrams commute:
$$
\parbox{.1in}{\xymatrix{\id_\mcS\ar[r]^{\eta}\ar[rd]_{\eta'}&T\ar[d]^F\\&T'}}
\hspace{.5in}\text{and}\hspace{.5in}
\parbox{.1in}{\xymatrix{T^2\ar[r]^\mu\ar[d]_{F^2}&T\ar[d]^F\\(T')^2\ar[r]_{\mu'}&T'.}}
$$

\end{definition}

\begin{remark}\label{rmk:upshot}

An important upshot of Proposition \ref{prop:Kls a functor} is the following. For any category $\mcC$ and morphism of monads $f\taking\top\to\top'$ we have a functor $\mcC\bkls{f}\taking\mcC\bkls{\top}\to\mcC\bkls{\top'}$. Thus any $\top$-instance on $\mcC$ can be transformed via $f$ into a $\top'$-instance on $\mcC$.

\end{remark}

\begin{definition}\label{def:category of monads}

A monad $\top=(T,\eta,\mu)$ is called {\em finitary} if the functor $T$ is determined by its values on finite sets, in the following sense. Let $X\in\Ob(\Set)$ be a set and let $\Fin{/X}$ denote the category whose objects are finite subsets of $X$ and whose morphisms are functions over $X$. Note that for each object $f\taking Y\to X$ in $\Fin_{/X}$ there is an induced map $TY\To{f}TX$, so we obtain a map $$M_X\taking\colim_{\parbox{.6in}{\tiny\begin{center}$f\taking Y\to X$\end{center}\begin{center}$\in\Fin_{/X}$\end{center}}}(TY)\To{\hsp} TX$$ Then $\top$ is finitary if the map denoted $M_X$ is a bijection for every $X\in\Ob(\Set)$.

The {\em category of finitary monads on $\Set$}, denoted $\Monad_\Set$, has finitary monads as objects and morphisms of monads (as in Definition \ref{def:morphisms of monads}) as morphisms. 

\end{definition}

\begin{proposition}\label{prop:Kls a functor}

A morphism of finitary monads induces a functor between their Kleisli categories. In other words there a functor $\Kls{-}\taking\Monad_\Set\to\Cat$.

\end{proposition}

\begin{proof}

This is straightforward.

\end{proof}

\subsection{Examples of transformations}\label{sec:examples of transformations}

In this section we write down several simple examples of morphisms of monads. In a few of these we are explicit, but we quickly move to a more colloquial style, assuming that any reader with sufficient interest and background can fill in the details for him or herself.

\subsub{Universals}

\begin{example}

The initial object in $\Monad_\Set$ is $\id_\Set$, whose instances are ordinary atomic database instances. Given any monad $\top$, there is a unique morphism of monads $\id_\Set\to\top$. As in Remark \ref{rmk:upshot}, there is a unique formula to convert any atomic instance into a $\top$-instance. 

\end{example}

\begin{example}

The terminal object in $\Monad_\Set$ is the Unlinked monad from Section \ref{sec:unlinked}. Given any monad $\top$, there is a unique morphism $\top\to\text{Unlinked}$. As in Remark \ref{rmk:upshot}, given any database instance on $\mcC$, be it atomic or not, one can forget all the foreign key information and be left with an unlinked instance. 

\end{example}

\subsub{Forgetting structure}

\begin{example}[Distributions to subsets]

Recall the Subset monad $\mcP$ and the Distribution monad $\Dist$ from Sections \ref{sec:subsets} and \ref{sec:distributions}, and let $X$ be a set. Recall that the support of a distribution $p\taking X\to[0,1]$ is the subset $\Supp(p)=\{x\in X\|p(x)\neq 0\}\ss X$. This notion of support induces a morphism of monads $\Dist\to\mcP$.

\end{example}

\begin{example}[Multisets to subsets]

Recall the Subset monad $\mcP$ and Multiset monad from Sections \ref{sec:subsets} and \ref{sec:multiset}, and let $X$ be a set. A multiset in $X$ can be conceived as a function $Y\to X$, and its image is a subset of $X$. By this process one obtains a morphism of monads $\text{Multiset}\to\mcP$.

\end{example}

\begin{example}[Lists to multisets]

Recall the List and Multiset monads from Sections \ref{sec:list} and \ref{sec:multiset}, and let $X$ be a set. For each natural number $n\in\NN$ we have a function $X^n\to X^n/\sim$ that forgets the order of $n$-element lists. This induces a morphism of monads $\List\to\text{Multiset}$.

\end{example}

\begin{example}

Recall the Atomic monad $\id_\Set$, the List monad, the non-empty list monad $\List_+$, and the Maybe monad $\Exc_\singleton$ from Sections \ref{sec:list} and \ref{ex:atomic and maybe}. There is a morphism which could be called the ``first element, if it exists" map $\List\to\Exc_\singleton$. Similarly, there is a ``first element" map $\List_+\to\id_\Set$.

\end{example}

\subsub{Tunable monads}

As explained in Section \ref{sec:tunable}, a tunable monad is a pair $(\mcI,P)$ where $P\taking\mcI\to\Monad_\Set$. For every object $i\in\Ob(\mcI)$ we have a monad $P(i)$ and for every arrow in $\mcI$ we have a morphism of monads. By Proposition \ref{prop:Kls a functor}, we can compose with the functor $\Kls{-}\taking\Monad_\Set\to\Cat$. So for any morphism $f\taking i\to i'$ in $\mcI$ and any schema $\mcC$ we have a functor $\mcC\bkls{P(i)}\to\mcC\bkls{P(i')}$ that functorially converts $P(i)$-instances into $P(i')$-instances.

In Section \ref{sec:tunable} we discussed the monads $\Exc_E$ and $\Inp^U$, for $E,U\in\Fin$, and $M$-annotated for monoids $M$. The value in these is found in the fact that as time goes on and the model evolves, the database architect may need to change parameters of the schema with minimal disturbance to users. For example, if at some point the architect wants to add a new sort of exception, or collapse two kinds of exception into one, he or she can do that by finding a function $E_{\text{old}}\to E_{\text{new}}$ from the old exception set to the new, and it will induce a functor that transforms databases instances with the old set of exceptions into instances with the new set.

\subsub{Others}

\begin{example}[Simulation]

Recall the Computation monad and the Experiments monad from Examples \ref{ex:computation} and \ref{ex:experiment}. We could imagine that simulation is a morphism of monads from the latter to the former, converting a description of an experiment into a computation.

\end{example}

\begin{example}[Programs take time]

Recall the Time-delay monad and the Computation monad from Examples \ref{ex:time delay} and \ref{ex:computation}. Counting the number of clock cycles induces a morphism of monads from the latter to the former.

\end{example}

\section{Future work}\label{sec:future}

The above work can be made far more flexible if we allow the choice of monad to vary over the schema. This way, some columns can be nullable and others not, or we could allow for lists in some areas of the schema and not in others. We will tackle this in an upcoming paper. It would also be interesting to consider how these variable monads and their associated instances would behave under change of schema functors $F\taking\mcC\to\mcD$. We also plan to investigate whether our work here can be nicely integrated with the ideas of \cite{BNT}, \cite{Gru}, \cite{LT}, and \cite{Maj}, in which one uses monads to handle collections. Finally, it seems fruitful to explore how the coincidence of (\ref{dia:riddle}) relates to Leinster's definitions of $\top$-multi-categories and operads in \cite[Chapter 4]{Le1}.

\bibliographystyle{amsalpha}

\end{document}